\documentclass[letterpaper, 10 pt, journal, onecolumn]{ieeeconf} 
\IEEEoverridecommandlockouts                   
\usepackage{graphics} 
\usepackage{times} 
 \usepackage{amsmath}
\usepackage{amssymb}
\usepackage{amsthm}
\usepackage{mathtools}
\usepackage{lipsum}
\usepackage{booktabs}

\usepackage{tikz}
\usetikzlibrary{quotes}
\usetikzlibrary{bayesnet}
\usetikzlibrary{shapes,arrows}
\definecolor{mcomment}{RGB}{0,153,76}

\title{\LARGE \bf
Pricing for Multi-modal Pickup and Delivery Problems with Heterogeneous Users}

\author{Mark Beliaev, Negar Mehr, Ramtin Pedarsani
	\thanks{M. Beliaev is with Graduate School of Electrical \& Computer Engineering, University of California Santa Barbara, Santa Barbara, CA, USA.}
	\thanks{N. Z. Mehr is with the Faculty of Aerospace Engineering, University of Illinois at Urbana-Champaign, Champaign, IL, USA.}
    \thanks{R. Pedarsani is with the Faculty of Electrical \& Computer Engineering, University of California Santa Barbara, Santa Barbara, CA, USA.}}

\newtheorem{thm}{Theorem}
\newtheorem{remark}{Remark}
\newtheorem{definition}{Definition}
\newtheorem{prop}{Proposition}
\newtheorem{cor}{Corollary}

\newcommand{\order}{i}
\newcommand{\Orders}{\mathcal{I}}
\newcommand{\mode}{j}
\newcommand{\Modes}{\mathcal{J}}
\newcommand{\numModes}{J}
\newcommand{\setmodes}{\{1,\ldots,\numModes\}}

\newcommand{\xij}{x_{\order,\mode}}
\newcommand{\xijp}{x_{\order',\mode'}}
\newcommand{\xijpp}{x_{\order'',\mode''}}
\newcommand{\flow}{x_\order}
\newcommand{\nashflow}{\tilde{x}_\order}
\newcommand{\setflow}{\{\xij\}_{\mode\in\Modes}}
\newcommand{\setnflow}{\{\tilde{x}_{\order,\mode}\}_{\mode\in\Modes}}

\newcommand{\tax}{\tau}

\newcommand{\tij}{\tax_{\order,\mode}}

\newcommand{\settax}{\{\tij\}_{\mode\in\Modes}}

\newcommand{\lat}{\ell}
\newcommand{\lij}{\lat_{\order,\mode}}
\newcommand{\lijp}{\lat_{\order',\mode'}}
\newcommand{\lijpp}{\lat_{\order'',\mode''}}
\newcommand{\setlat}{\{\lij\}_{\mode\in\Modes}}
\newcommand{\tlat}{L}

\newcommand{\pop}{\alpha}
\newcommand{\user}{a}

\newcommand{\cost}{p}
\newcommand{\cija}{\cost_{\order,\mode}^{\user}}
\newcommand{\cixa}{\cost_{\order,x_\order(\user)}^{\user}}

\newcommand{\opcj}{c_{\mode}}
\newcommand{\bopcost}{C}

\newcommand{\loc}{d}
\newcommand{\rest}{r}

\newcommand{\traveltime}{t}
\newcommand{\pickuptime}{u}
\newcommand{\service}{s}

\newcommand{\bij}{\beta_{\order,\mode}}
\newcommand{\slij}{\service_{\order,\mode}}
\newcommand{\plij}{\pickuptime_{\order,\mode}}
\newcommand{\tlij}{\traveltime_{\order,\mode}}

\newcommand{\drate}{\mu}
\newcommand{\minutes}{k}
\newcommand{\couriercap}{N}

\begin{document} 
	
	\maketitle
	\begin{abstract}
	In this paper, we study the pickup and delivery problem with multiple transportation modalities, and address the challenge of efficiently allocating transportation resources while price matching users with their desired delivery modes. More precisely, we consider that orders are demanded by a heterogeneous population of users with varying trade-offs between price and latency. To capture how prices affect the behavior of heterogeneous selfish users choosing between multiple delivery modes, we construct a congestion game taking place over a form of star network, where each source-sink pair is composed of parallel links connecting users with their preferred delivery method. Using the unique geometry of this network, we prove that one can set prices explicitly to induce any desired network flow, i.e, given a desired allocation strategy, we have a closed-form solution for the delivery prices. We conclude by performing a case study on a meal delivery problem with multiple courier modalities using data from real world instances.     
	\end{abstract}
	\section{Introduction}\label{sec: Introduction}
	As the world continues to integrate with digital technology, we become more reliant on e-commerce services such as food delivery and ride-hailing. The global food delivery market has seen exponential growth, with the most mature markets becoming four to seven times larger from 2018 to 2021~\cite{ahuja2021ordering}. In 2022, Uber reported a $19\%$ year-over-year increase in online bookings, marking a daily average of $23$ million trips on their platform~\cite{business_wire_2023}. Despite this growth, many pickup and delivery services operate under low profit margins due to high driver wages~\cite{CDC_pravin}.\par 
	
	To fulfill market demands and mitigate these costs, recent efforts have been made to introduce autonomous transportation methods for food delivery and ride hailing, such as delivery drones, electric vertical takeoff and landing (eVTOL) aircrafts, and sidewalk autonomous delivery robots (SADRs)~\cite{MOSHREFJAVADI2021114854,robot_starship}. For example, Archer Aviation Inc. has recently revealed their plans to provide passengers with an eVTOL aircraft travel option between O’Hare International Airport and Vertiport Chicago that takes roughly $10$ minutes, a trip which can take upwards of an hour or more during rush hour traffic~\cite{archer}.\par 
	
	As these modalities are introduced, it is important for service providers to develop new resource allocation strategies that efficiently utilize emergent transportation modalities, while coinciding with customer preferences. With the advent of urban air mobility, recent research has studied demand modeling, operations, and integration with existing infrastructure~\cite{GARROW2021103377}. For example, surveys examining commuter preferences regarding transportation have found significant heterogeneity in individuals' value of time, and that the median value of time for air taxis is larger compared to other modalities~\cite{binder_2018,garrow_2019}. As the inclusion of urban air mobility is predicted to disrupt urban transportation, it will be crucial for existing rideshare providers to adapt their pricing and allocation strategies. In light of these developments, we address the challenge of using customer preferences to set prices that efficiently allocate transportation resources amongst them.\par
	
	This paper examines the pickup and delivery problem with multiple transportation modalities, and demonstrates how one can achieve a desired allocation strategy for a set of orders by appropriately setting prices for each modality. Specifically, we consider orders demanded by a heterogeneous population of users with varying trade-offs between price and latency. This problem is analogous to a congestion game taking place over a form of star network: one central \textit{source} node is used to connect a set of \textit{sink} nodes, where for each source-sink pair, there is a directed graph composed of parallel links between the source and sink node. This way, we can use the star network to represent a central delivery system responsible for connecting users placing some order, with their preferred transportation modality. Note that by using parallel links to represent the modalities, we are not concerned with routing individual vehicles, and instead focus on allocating transportation resources. We illustrate this congestion game in Figure~\ref{fig:front_fig}, depicting a meal delivery problem where customers at different locations place orders for food via the available delivery modes. This unique network structure enables us to show that we can explicitly define prices to induce any desired network flow, i.e, given a desired allocation strategy, we have a closed-form solution for the delivery prices.\par
	
	\begin{figure}[!t]
		\caption{We represent the pickup and delivery problem as a congestion game played over a star network. {Each source-sink pair is} denoted by $\order\in\Orders$, which can be viewed as a population of users at some location demanding a particular order at a certain rate. Each source-sink pair is connected by a set of parallel edges $\mode\in\Modes$, which can be viewed as the set of delivery modes the users choose from. Note that we are not concerned with how the couriers are routed to the pickup or delivery location, and instead focus on how we allocate the different delivery modes for each order. Specifically, our goal is to induce an optimal allocation of transportation modalities by appropriately setting prices for each order-modality pair.}
		\includegraphics[width=0.8\columnwidth]{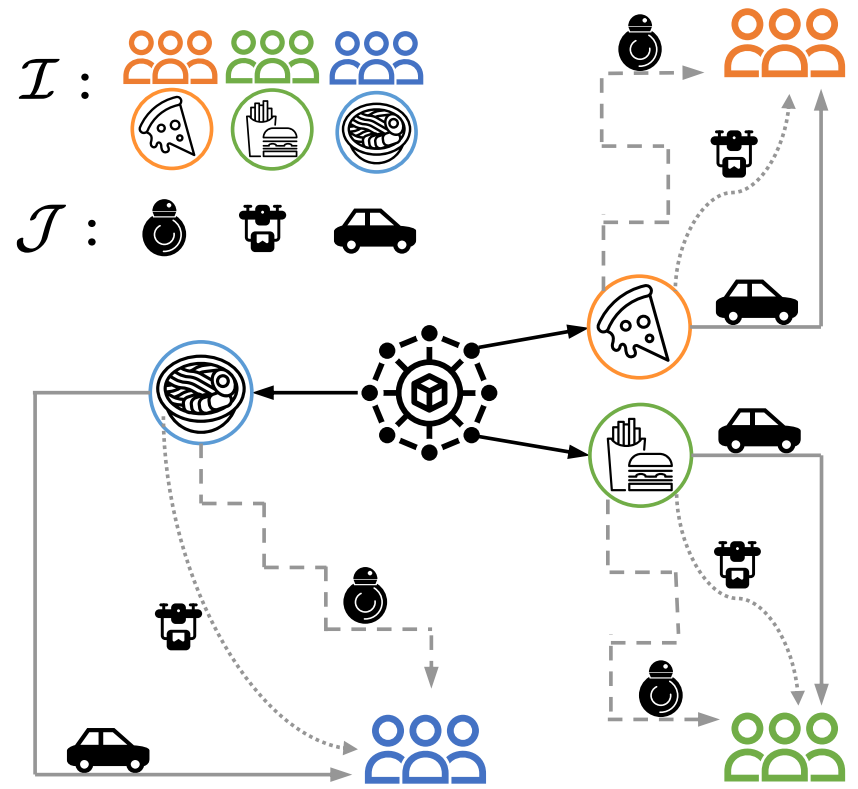}
		\label{fig:front_fig}
	\end{figure}
	
	The main contributions of this work are:
	\begin{itemize}
		\item We construct a congestion game that captures how prices affect the behaviour of heterogeneous selfish users choosing between multiple delivery modes.
		\item Building on results from prior works, we prove that in these settings, the set of prices can be explicitly defined for any desired network flow.
		\item We demonstrate our results with a case study on a meal delivery service with multiple courier modalities, using real world instances provided by Grubhub~\cite{reyes2018meal}. 
		\item Under additional assumptions, we extend our results by allowing users to have varying trade-offs for each modality. We demonstrate this with a case study on a taxi service with urban air transportation to and from the O'Hare International Airport, using data provided by the city of Chicago~\cite{TNP2023}.
	\end{itemize}
	
	\section{Related Work}
	The application of emerging transportation modalities such as unmanned aerial vehicles or drones has drawn a lot of attention. Many works look at how drones can be utilized in logistic operations such as delivery systems~\cite{beliaev2022}, urban air taxi~\cite{airtaxi}, on-demand meal delivery~\cite{LIU20191}, as well as many other applications~\cite{MOSHREFJAVADI2021114854}. Other works look at safety verification for dynamical systems utilizing drones to account for factors such as collision avoidance~\cite{coogan_2022} and schedule feasibility~\cite{coogan_2021}. The pickup and delivery vehicle routing problem with drones has also been considered by some, where mixed integer linear programming models are used to find routing solutions for optimizing various objectives~\cite{pdp_routing_1,pdp_routing_2}. Unlike these works, our research lies in the broader field of congestion games, specifically building on previous works that consider pricing in non-atomic congestion games.\par 
	
	Congestion games aim to allocate traffic over transportation networks represented by graphs, where each road corresponds to an edge with a latency function representing the travel time experienced by users on that edge~\cite{dafermos1969traffic}. In these settings, one aims to find the optimal network flow that minimizes a social cost, such as the aggregate latency experienced by all users. However, if we assume that users are self-interested and choose their routes selfishly by minimizing their individual latency, the resulting flow follows a network equilibrium~\cite{wardrop1952, sheffi1984}. One area of research is focused on categorizing the trade-off in social cost between the optimal network flow and the equilibrium network flow~\cite{roughgarden2005, lazar_tac, Lazar_Coogan}. Many works specifically look at how tolling can be used to price network edges such that the equilibrium network flow corresponds to the desired optimal flow~\cite{dafermos1973, roughgarden2003b, marden2017}. In our work, we make the distinction that users are heterogeneous in their trade-off between price and time.\par
	
	While in the homogeneous case it has been long known that marginal cost pricing can guarantee that the equilibrium flow equals the optimal flow~\cite{beckmann1956studies}, this strategy does not hold for heterogeneous populations. More recent research has demonstrated that for directed graphs with one source-sink pair, optimal tolls exist and can be found by solving a polynomial size set of linear inequalities, given that the number of users in the heterogeneous population is finite~\cite{roughgarden2003a}. In this seminal paper, it was assumed that the model was nonatomic, meaning that each user corresponded to an infinitesimal unit of flow, and inelastic, meaning that the demand could not change as a function of the road parameters. Following this work, others have improved the result by considering multicommodity networks~\cite{Karakostas2004EdgePO,Fleischer2004}, allowing user demand to be elastic~\cite{Karakostas2006EdgePO}, and addressing the atomic setting~\cite{Fotakis2010OnTE}. In this paper, we keep the assumption of a nonatomic model with inelastic demands, but consider a graph structure which is unique to the pickup and delivery problem considered. By exploiting this graph structure, we can define prices explicitly to induce any desired network flow without limiting it to an optimal flow. Whereas prior works directly use Linear Program (LP) formulations to find edge prices in general directed graphs, our theoretical results imply that one can first find path prices combinatorially to simplify the LP formulation.\par 
	
	The rest of the paper is organized as follows. In the subsequent Section~\ref{sec: Problem Setup}, we formally introduce the problem setting and show how it is analogous to a congestion game. Following this, in Section~\ref{sec: Theoretical Framework} we describe our theoretical result in the general framework of the aforementioned congestion game, and in Section~\ref{sec: Model}, we describe the specific framework that is used to model the pickup and delivery problem and find the optimal allocation strategy. We go on to apply our theoretical results on this framework in Section~\ref{sec: Case Studies}, which consists of our two case studies using the public Grubhub dataset~\cite{reyes2018meal} for meal delivery and the Chicago Transportation Network Providers dataset~\cite{TNP2023} for taxi services. Lastly, we conclude our work in Section~\ref{sec: Conclusion}, listing potential avenues for improvement and further research.\par
	
	\section{Problem Formulation}\label{sec: Problem Setup}
	We model our pickup and delivery problem using a static system, where during a given time interval\footnote{Without loss of generality, we can define the time interval during which the orders $\Orders$ are demanded as one hour, using the same unit of time for all variables and constants throughout our formulation.}, there is a set of orders $\Orders$ demanded by a population of users. We consider that each order originates from a unique neighborhood composed of a heterogeneous population represented by the interval $[0,1]$, where each point $\user\in[0,1]$ is a non-cooperative and infinitesimal unit referred to as a user. We sort these users by money sensitivity, where in general, we can view $\pop_\order:[0,1]\rightarrow(0,\infty)$ as an unbounded, non-decreasing function representing the trade-off between pisc corresponding to order $\order$. Thus, when placing an order $\order\in\Orders$, each user chooses one of the $\numModes$ delivery modes $\mode\in\Modes:\{1,\ldots,\numModes\}$ based on latency $\lij$, dollar price $\tij$, and their money/time valuation $\pop_\order(\user)$. We assume users placing order $\order\in\Orders$ have inelastic demands, i.e., they will not switch their demand to a different order and will always choose one of the $\numModes$ delivery modes. Our goal is then to find the set of delivery prices which would induce some desired allocation of users between the delivery modes $\mode\in\Modes$ for each order $\order\in\Orders$.\par
	
	\begin{table}[!h]
		\centering 
		\caption{Notations}
		\centering
		\begin{tabular}{r c p{10cm} }
			\toprule
			$\Orders$ $|$ $\order\in\Orders$ & $\triangleq$ & Set of orders $|$ individual order (source-sink pair) \\
			$\Modes$ $|$ $\mode\in\Modes$ & $\triangleq$ & Set of modes $|$ individual mode (edge)\\
			$\numModes$ & $\triangleq$ & Number of modes, i.e., the cardinality of $\Modes$ \\
			$\user\in[0,1]$ & $\triangleq$ & A non-cooperative and infinitesimal unit referred to as a user \\
			$\pop_{\order}:[0,1]\rightarrow(0,\infty)$ & $\triangleq$ & Function representing the trade-off between price and time for users placing order $\order$.\\
			$\lij$ & $\triangleq$ & Latency: total time required to complete order $\order$ using mode $\mode$ \\
			$\tij$ & $\triangleq$ & Dollar price of placing order $\order$ using mode $\mode$ \\
			$\xij$ & $\triangleq$ & Flow of users on edge $\mode$ corresponding to source-sink pair $\order$ \\
			$\flow:[0,1]\rightarrow\Modes$ & $\triangleq$ & Function representing flow over the edges $\Modes$ for source-sink pair $\order$\\
			$\cija$ & $\triangleq$ & User cost (hours) $\user\in[0,1]$ assigns to edge $\mode$ for source-sink pair $\order$ \\
			$\slij$ $|$ $\tlij$ $|$ $\plij$& $\triangleq$ & Service time $|$ travel time $|$ pickup time (order $\order$, mode $\mode$)\\
			$\rest_\order$ $|$ $\loc_\order$ & $\triangleq$ & Pickup location $|$ drop-off location (for order $\order$) \\
			$\couriercap_\mode$ & $\triangleq$ & the total number of vehicles for mode $\mode$ \\
			$\drate_\mode$ & $\triangleq$ & Order completion rate for mode $\mode$ in units of orders per hour \\
			$\rho_\mode$  $|$ $\bar{\rho}_\mode$ & $\triangleq$ & Utilization of mode $\mode$ $|$ upper bound on utilization of mode $\mode$  \\
			$\bij$ & $\triangleq$ & Portion of available couriers distributed around pickup location $\rest_\order$ such that their travel times are uniform in $[0,\minutes_\mode]$ \\
			$\minutes_\mode$ & $\triangleq$ & Constant unit of time used for computing $\bij$ \\
			$\opcj$ & $\triangleq$ & Cost in dollars for completing one order with mode $\mode$ \\
			$\bar{c}_\mode$ & $\triangleq$ & Cost in dollars per hour for operating mode $\mode$ \\
			$\bopcost$ & $\triangleq$ & Cost in dollars per hour for operating the entire system \\
			\bottomrule
		\end{tabular}
		\label{tab:TableOfNotation}
	\end{table}
	
	This problem is analogous to a congestion game played over a star network as portrayed in Figure~\ref{fig:front_fig}, where each order $\order\in\Orders$ corresponds to a source-sink pair connected by a set of parallel edges $\Modes$ representing the different delivery options. Each source-sink pair $\order\in\Orders$ has an associated demand of traffic flow at the sink which represents the population of users $\user\in[0,1]$ requesting deliveries. Although we model this flow demand with the unit interval to simplify notation, we can allow for an arbitrary demand $r_\order$ at each source-sink pair $\order$. The edge corresponding to modality $\mode\in\Modes$ for source-sink pair $\order\in\Orders$ has a congestion dependent latency $\lij$, which represents the time needed to complete the order, and a price issued to control congestion $\tij$, which represents the dollar price paid by the user, both of which are assumed to be nonnegative. Note that when we drop index $\mode$ from the notation of terms like $\lij$ by writing $\lat_\order$, we refer to the set of latencies $\setlat$ over all the edges $\Modes$ for a given source-sink pair $\order$.\par 
	
	With this approach, we can view network flow as an allocation of users over the delivery modes. To represent such allocation strategies, we define $0\leq \xij\leq1$ as the flow of users on edge $\mode\in\Modes$ corresponding to source-sink pair $\order\in\Orders$, where $\sum_{\mode\in\Modes}\xij=1$ must be satisfied. More precisely, for each source-sink pair $\order$ we view this flow as a Lebesgue-measurable function $\flow:[0,1]\rightarrow\Modes$ which corresponds to a flow over the edges $\setflow$. We use notations $x=\{\xij\}_{\order\in\Orders,\mode\in\Modes}$ and $\tax=\{\tij\}_{\order\in\Orders,\mode\in\Modes}$ to denote the entire set of edge flows and edge prices, respectively. As we will later show in Section~\ref{sec: Model} when defining the specific optimization problem, we can use $x$ as a decision variable to find an optimal allocation strategy for a given objective, and explicitly define prices $\tax$ that induce this desired strategy. For now, we continue to detail how latency and user equilibrium are considered in our framework.\par
	
	\subsection{Congestion}
	We first describe the congestion element of our framework, namely, the latency function defined for each edge. Specifically, we assume that each edge $\mode\in\Modes$ corresponding to source-sink pair $\order\in\Orders$ has a nonnegative and continuous latency $\lij$ as a function of the entire network flow $x$. Each latency function $\lij$ describes the time it takes for an order $\order$ delivered by modality $\mode$ to arrive at the customer's location from the moment it was placed. We note that in order to claim our main theoretical result, we do not need any further restrictions on the latency functions $\lij$. We leave further discussion regarding latency to Section~\ref{sec: Model}, where we model latency using concepts from queuing theory for our application. Until then, we stick with the aforementioned assumptions and simply use notation $\lij(x)$ when defining edge latency.\par
	
	\subsection{User Equilibrium}
	We are now ready to discuss how users choose between the different delivery modes. When confronted with a set of prices $\tax_\order$ and latencies $\lat_\order$ for the varying edge options $\mode\in\Modes$, user $\user\in[0,1]$ will choose the edge with the smallest cost $\lij(x)+\pop_\order(a)\tij$. Essentially, every source-sink pair $\order\in\Orders$ corresponds to its own nonatomic game in which users $\user\in[0,1]$ choose between the $\mode\in\Modes$ pure strategies available. The non-cooperative behaviour of users results in a Nash equilibrium, which is a stable point where no user has an incentive to unilaterally alter their chosen strategy. Specifically, we let $\cija(x,\tax_\order)=\lij(x)+\pop_\order(\user)\tij$ represent the evaluation user $\user\in[0,1]$ assigns to edge $\mode$ for source-sink pair $\order$.
	\begin{definition}
		For a given source-sink pair $\order\in\Orders$, we call the flow $\flow:[0,1]\rightarrow\Modes$ an equilibrium or Nash flow for instance $(\pop_\order,\lat_\order,\tax_\order)$ if for any user $\user\in[0,1]$ and edge $\mode\in\Modes$:
		\begin{equation}\label{eq:Nash_condition}
			\cixa(x,\tax_\order)\leq \cija(x,\tax_\order).
		\end{equation}
	\end{definition}
	The existence of such Nash flows is a well known and a general result~\cite{schmeidler_equilibrium}. 
	\begin{prop}\label{prop: ex_nash}
		For a given source-sink pair $\order\in\Orders$, any instance $(\pop_\order,\lat_\order,\tax_\order)$ admits a Nash flow $\flow:[0,1]\rightarrow\Modes$ satisfying Eq.~\eqref{eq:Nash_condition}. 
	\end{prop}
	\par
	We point out that the above Proposition requires not only the cost function $\cija$ to be nonnegative and continuous, but also the set of possible flows to be a nonempty, convex, and compact. This is trivially satisfied when considering only the demand, as done when presenting our theoreitcal results in Section~\ref{sec: Theoretical Framework}. However, we will need to make sure the set is nonempty when including a supply constraint, as done when formulating our optimization problem in Section~\ref{sec: Model}. 
	
	Note that in the above results pertaining to Nash equilibria, for each source sink-pair $\order\in\Orders$, we consider the flow $\flow$ independently from the entire network flow $x$, keeping the remaining flows constant. Since the latency $\lij(x)$ is assumed to be a function of the entire network flow $x$, one may require a network flow $x:\{\flow\}_{\order\in\Orders}$ for which all source sink pairs $\order\in\Orders$ exhibit Nash equilibrium under their corresponding flow $\flow:[0,1]\rightarrow\Modes$. We use the term \textit{stable allocation strategy} to encompass this notion, formally defining it below. 
	\begin{definition}\label{def:nash}
		For a given star network defined by the source-sink pairs $\order\in\Orders$ and edges $\mode\in\Modes$, we call the network flow $x:\{\flow\}_{\order\in\Orders}$ a stable allocation strategy for instance $(\pop,\lat,\tax)$ if for all source-sink pairs $\order\in\Orders$, the corresponding flow $\flow:[0,1]\rightarrow\Modes$ is an equilibrium flow satisfying Eq.~\eqref{eq:Nash_condition}.  
	\end{definition}
	As we show in the subsequent section, any network flow $x$ is a stable allocation strategy for some set of prices $\tax$.\par
	
	\section{Theoretical Results}\label{sec: Theoretical Framework}
	Before stating our main results, we need to elaborate on one more property of equilibrium flows that applies to individual source-sink pairs. Intuitively, we expect Nash flows to exhibit a structure where users $\user\in[0,1]$ close to $0$, who value time more than money, will choose an option with small latency but large price. Similarly, users further away from $0$ will choose an option with a relatively larger latency but a smaller price. Finally, users close to $1$ will choose an option with very large latency in order to pay a very small price. We encapsulate this notion below. 
	\begin{definition}\label{def:canon}
		For a given source-sink pair $\order\in\Orders$, a flow $\flow$ at Nash equilibrium is \textit{canonical} if:
		\begin{itemize}
			\item For any edge $\mode\in\Modes$, the users assigned to $\mode$ form a possibly empty or degenerate subinterval of $[0,1]$.
			\item If $\user_1<\user_2$, then $\lat_{\order,\flow(\user_1)}(x)\leq\lat_{\order,\flow(\user_2)}(x)$.
			\item If $\user_1<\user_2$, then $\tax_{\order,\flow(\user_1)}\geq\tax_{\order,\flow(\user_2)}$.
		\end{itemize}
	\end{definition}
	
	\begin{figure}[!h]
		\centering
		\begin{tikzpicture}[scale=7]
			\draw[-, thick] (0,0) -- (1,0);
			\foreach \x/\xtext in {0/0,0.3/$\user_{\mode-1}$,0.5/$\user_{\mode}$,0.7/$\user_{\mode+1}$,1/$1$}
			\draw[thick] (\x,0.5pt) -- (\x,-0.5pt) node[below] {\xtext};
			\draw (0.15,-1pt) node {$\ldots$};
			\draw (0.85,-1pt) node {$\ldots$};
			\draw[[-), ultra thick, blue] (0.3,0) -- (0.5,0);
			\draw[[-), ultra thick, red] (0.5,0) -- (0.7,0);
			\draw (0.4,3pt) node {$\tax_\mode$};
			\draw (0.4,1.5pt) node {$\lat_\mode$};
			\draw (0.6,3pt) node {$\tax_{\mode+1}$};
			\draw (0.6,1.5pt) node {$\lat_{\mode+1}$};
			\draw (0.5,3pt) node {$\geq$};
			\draw (0.5,1.5pt) node {$\leq$};
		\end{tikzpicture}
		\caption{A sketch depicting how a canonical Nash flow splits the population $\user\in[\user_0,\user_{\numModes}]$ into subintervals $[\user_{\mode-1},\user_\mode): x(\user)=\mode$, where $\user_0=0$, $\user_{\numModes}=1$, and $\mode\in\Modes$. Note that order $\order$ is left out from notation.} 
		\label{fig:proof intervals}
	\end{figure}
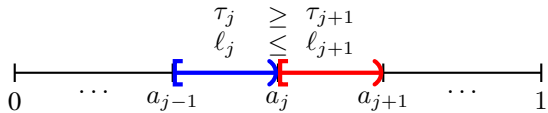
	
	In other words, a canonical Nash flow $\flow$ splits $[0,1]$ into at most $\numModes$ potentially degenerate sub intervals, inducing an ordering over the edges to which $\flow$ assigns users that is nondecreasing in latency and nonincreasing in prices. We portray this in Fig.~\ref{fig:proof intervals}. Using results from prior work which proposed this definition~\cite{roughgarden2003a}, we can state the following existence property, providing an independent proof of this proposition in \ref{sec: app_prop_proof} for completeness.
	
	\begin{prop}\label{prop: ex_canon}
		For a given source-sink pair $\order\in\Orders$, every instance $(\pop_\order,\lat_\order,\tax_\order)$ admits a canonical Nash flow. 
	\end{prop}
	
	With these properties, we can say that for a given source-sink pair $\order\in\Orders$ and instance $(\pop_\order,\lat_\order,\tax_\order)$, there exists a canonical Nash flow $\nashflow:[0,1]\rightarrow\Modes$. This canonical  Nash flow represents the flow $\setnflow$, where users in interval $[\user_{\mode-1},\user_\mode]\in[0,1]$ are routed on edge $\mode$ for some corresponding set $\user_0\leq\user_1\leq\ldots\leq\user_{\numModes}$, with $\user_0=0$ and $\user_{\numModes}=1$. In the pickup and delivery setting, we can assume that the delivery provider already has a set of flows $\setflow$ representing the desired allocation strategy for order $\order$, and wants to find a corresponding set of prices $\settax$ such that the induced equilibrium flow $\setnflow$ is equal to the desired flow. Building on top of the aforementioned results, we find a closed-form solution to this problem.
	\begin{thm}
		\label{thm:pricing condition}
		For a given source-sink pair $\order\in\Orders$, any desired flow $\setflow$ is an equilibrium flow for instance $(\pop_\order,\lat_\order,\tax_\order)$, where the set $\Modes:\setmodes$ orders the edges by non-decreasing latency, $\pop_\order:[0,1]\rightarrow(0,\infty)$ is a non-decreasing distribution function, $\lat_\order$ is the set of corresponding edge latencies, and $\tax_\order$ is the set of prices defined by:
		\begin{equation}\label{eq:main_result}
			\tij=\tax_{\order,\numModes}+\sum_{k=\mode}^{\numModes-1}\frac{\lat_{\order,k+1}-\lat_{\order,k}}{\pop_\order(\user_{k})}\qquad\forall\mode\in\Modes,
		\end{equation}
		where $\tax_{\order,\numModes}$ is any predefined price for the cheapest option. 
	\end{thm}
	\begin{proof}
		The proof strategy is as follows: 
		given the result of Proposition~\ref{prop: ex_canon} which states that every instance $(\pop_\order,\lat_\order,\tax_\order)$ admits a canonical Nash flow, we use the properties of canonical Nash flows along with a subset of the inequalities defined for Nash equilibrium in Eq.~\eqref{eq:Nash_condition} to show that for some desired flow $\setflow$ to be at equilibrium, there is only one set of valid prices $\tax_\order$. We complete the proof by showing that the corresponding set of prices $\tax_\order$ does indeed satisfy all of the inequalities defined in Eq.~\eqref{eq:Nash_condition}. The full proof is provided in \ref{sec: app_proof}.
	\end{proof}
	It follows directly that given any network flow $x:\{\flow\}_{\order\in\Orders}$ representing a desired allocation strategy over all orders, one can independently set prices $\tax:\{\tax_\order\}_{\order\in\Orders}$ for each source-sink pair to make $x$ a stable allocation strategy.
	\begin{cor}\label{corr: ex_network_flow}
		For a given star network defined by source-sink pairs $\order\in\Orders$ and edges $\mode\in\Modes$, any network flow $x:\{\flow\}_{\order\in\Orders}$ is a stable allocation strategy for instance $(\pop,\lat,\tax)$ when the set of prices $\tax$ is defined according to Eq.~\eqref{eq:main_result}.
	\end{cor}
	
	We note that under additional assumptions, the results of Theorem~\ref{thm:pricing condition} and Corollary~\ref{corr: ex_network_flow} can be extended to the setting where users have different trade-offs between price and latency $\pop_{\order,\mode}$ for different modes of transportation $\mode$ as well as orders $\order$.
	
	\begin{cor}\label{corr: varying modality}
		For a given source-sink pair $\order\in\Orders$ and instance $(\pop_\order,\lat_\order,\tax_\order)$, if a desired flow $\setflow$ is inducible under some equilibrium flow $\nashflow:[0,1]\rightarrow\Modes$ that routes users in interval $[\user_{\mode-1},\user_\mode]\in[0,1]$ on edge $\mode$ for some corresponding set $\user_0\leq\user_1\leq\ldots\leq\user_{\numModes}$, then the set of prices $\tax_\order$ must be defined by:
		\begin{equation}\label{eq:extended_main}
			\tij=\frac{\pop_{\order,\mode+1}(\user_\mode)}{\pop_{\order,\mode}(\user_\mode)}\tax_{\order,\mode+1} + \frac{\lat_{\order,\mode+1}-\lij}{\pop_{\order,\mode}(\user_\mode)} \qquad \forall\mode\in\{1,\ldots,\numModes-1\},
		\end{equation}
		where $\user_0=0$ and $\user_{\numModes}=1$, $\pop_{\order}$ is a set of non decreasing functions $\pop_{\order,\mode}:[0,1]\rightarrow(0,\infty)$ for $\mode\in\setmodes$ such that given $\user<\user'$, $\frac{\pop_{\order,\mode+1}(\user)}{\pop_{\order,\mode}(\user)}\geq \frac{\pop_{\order,\mode+1}(\user')}{\pop_{\order,\mode}(\user')}$ for all $\mode\in\{1,\ldots,\numModes-1\}$, $\lat_\order$ is the set of corresponding edge latencies, and $\tax_{\order,\numModes}$ is any predefined price for the cheapest option. It follows directly that if the network flow $x:\{\flow\}_{\order\in\Orders}$ is a stable allocation strategy for instance $(\pop,\lat,\tax)$, then the set of prices $\tax$ is defined by Eq.~\eqref{eq:extended_main}.
	\end{cor}
	
	The proof is provided in \ref{sec: app_corr_proof}, where we can no longer rely on Proposition~\ref{prop: ex_canon} that allows us to utilize the properties of canonical Nash flows. Due to this, we only consider well behaved equilibrium flows which divide the users into $\numModes$ potentially degenerate sub intervals, inducing an ordering over the edges which preserves the assumption placed on the distribution functions $\pop_{\order,\mode}$. Intuitively, this assumption requires the modalities $\setmodes$ ordered by their relative luxury, since users $\user\in[0,1]$ closer to $0$, who have a greater value of time overall, will have a proportionally higher value of time for more luxurious modes. Hence, Corollary~\ref{corr: varying modality} tells us that if the desired flow $\setflow$ can be induced by such an equilibrium flow, then the prices $\settax$ must follow Eq.~\ref{eq:extended_main}. Although this is a weaker result compared to Theorem~\ref{thm:pricing condition}, it can be strengthened and used in applicable settings as we will demonstrate in our second case study provided in Section~\ref{sec: Case Studies}.\par
	
	Lastly, we point out that so far we have not made any claims regarding the uniqueness of equilibrium flows. While the uniqueness of Nash flows for networks composed of parallel links is a known result (see Theorem 1 in~\cite{orda1993} or Proposition 3.3 in~\cite{milchtaich2000}), it requires additional assumptions: convexity and strict monotonicity of cost functions $\cija$ with respect to flow $\xij$. To keep our results in Theorem~\ref{thm:pricing condition}, as well as Corollaries~\ref{corr: ex_network_flow}and~\ref{corr: varying modality}, general for any arbitrary latency functions $\lij$, we forego making such restrictions. However, we note that the latency function $\lij$ used in our case studies satisfies the aforementioned requirements. Furthermore, due to the results of Proposition~\ref{prop: ex_canon}, any two allocation strategies $x$ and $x'$ that are stable for instance $(\pop,\lat,\tax)$, must induce the same ordering over the edges $\mode\in\Modes$ that is nondecreasing in latency $\lij$, and nonincreasing in prices $\tij$, for all source-sink pairs $\order\in\Orders$.

	\section{Pickup and Delivery Problem }\label{sec: Model}
	To show the usability of our model, we apply our theoretical framework to the pickup and delivery problem with multiple courier types. Our goal is to find the optimal allocation strategy with respect to some objective, where we will use Theorem~\ref{thm:pricing condition} to set the prices which induce this desired strategy. Our objective will be to find the optimal values of $x$ which minimize the expected latency over all orders:
	\begin{equation}\label{eq: expected ctd}
		\tlat(x)=\frac{1}{|\Orders|}\sum_{\order\in\Orders}\sum_{\mode\in\Modes}\lij(x)\xij.
	\end{equation}
	Using the model and optimization problem developed in this section, we will perform case studies on  both a meal delivery service and a taxi service in the following Section~\ref{sec: Case Studies}.  Before we set up and solve this optimization problem, we first specify how latency is measured, and how cost is accounted for.\par
	
	\subsection{Latency Model}
	We begin by characterizing each order $\order\in\Orders$ by a $2$--tuple $\langle\rest_\order,\loc_\order\rangle$, consisting of a pick-up and drop-off location, respectively. We would like our system to model the time it takes to complete a customer's order from the moment it was placed. We refer to this as the latency $\lij$ for order $\order\in\Orders$ and modality $\mode\in\Modes$, computing it as:
	\begin{equation}\label{eq:click-to-door}
		\lij(x)=\slij+\tlij+\plij(x).
	\end{equation}
	Essentially, the above Eq.~\eqref{eq:click-to-door} splits the latency $\lij$ into three components: service time $\slij$, travel time $\tlij$, and pickup time $\plij$ for modality $\mode$ of order $\order$.\par 
	
	In total, to compute the latency $\lij$ of an order, we account for how long it takes a courier to arrive at the designated pickup location $\plij$, the travel time between the pickup and drop-off locations $\tlij$, and the service time required $\slij$. We view the service time $\slij$ as a constant representing the time spent at the pickup and drop-off locations when completing order $\order$ using modality $\mode$. Some examples of this include parking for vehicle couriers, landing for aerial couriers, loading, and unloading. Similarly, we define the travel time $\tlij$ as the time it takes to physically travel between pickup $\rest_\order$ and drop-off $\loc_\order$ locations using modality $\mode$. The travel time $\tlij$ between locations can be pre-computed separately for each modality $\mode$ and order $\order$ using some known functions. Lastly, we view the pick-up time $\plij$ as the time it takes for a courier of modality $\mode$ to arrive at pick-up location $\rest_\order$. Unlike the other two components, the time required for pickup $\plij$ should depend on the availability of couriers captured by our decision variable $x$, as well as the expected travel time between the pickup location and nearest available courier.\par 
	
	To account for the availability of couriers, we use the concept of server utilization from queuing theory. Specifically, we use the $M/M/c$ queue as an approximate model for the availability of couriers since we can obtain closed form formulas for the average order arrival and order completion rates. For a given modality $\mode$, we set $c$ to be the total number of couriers $\couriercap_\mode$, approximate the rate at which users are placing orders as $\sum_{\order\in\Orders}\xij$, and define the rate at which an order is completed by these types of couriers as $\drate_\mode$. Note that we can define the order completion rate $\drate_\mode$ as a constant provided by historical data, or estimate it using the parameters of our problem instance as we will do in the case studies following. Drawing these analogies allows us to define the utilization $\rho_\mode$ of our queuing system for couriers of modality $\mode$ as:
	\begin{equation}\label{eq:system utilization}
		\rho_\mode = \frac{\sum_{\order\in\Orders}\xij}{\couriercap_\mode\drate_\mode}.
	\end{equation}
	\par
	In our regime of interest, the rate of order arrivals is magnitudes larger than the rate of order completions, and hence the number of available couriers $c$ needs to be large. Using the $M/M/c$ latency function, one can show that in this regime of interest, the time spent waiting for an available server is negligible unless we are close to the capacity limit~\cite{queueing_theory}. For example, given a system with $c=50$ servers and a demand of $100$ requests per hour, when the server utilization is high at $\rho=0.99$, the average time spent in the system is $84$ minutes, with $55$ minutes in the queue. Once we lower the utilization to $\rho=0.9$, the average time spent in the system is $29$ minutes, with only $2$ minutes spent in the queue. This means that the expected waiting time for a courier to be available is relatively small compared to the latency required to complete the order, given that the utilization parameter $\rho_\mode$ is below a reasonable threshold. Thus, to make sure that customers are not experiencing long wait times for couriers to respond, we can upper-bound the utilization parameter $\rho_\mode$ for all courier types, and ignore the effect of varying availability.\par
	
	To model the time a courier must spend traveling to the pick-up location $\rest_\order$, we take a probabilistic approach by calculating the expected travel time of the nearest available courier. Specifically, we assume that for modality $\mode$, some portion $\bij\in(0,1]$ of available couriers are distributed around the pick-up location $\rest_\order$ such that their travel times are uniform in $[0,\minutes_\mode]$. Note that we can choose $\minutes_\mode$ as some constant unit of time from which $\bij$ is estimated based on the pick-up location and modality. Since we know that the expected number of available couriers will be $(1-\rho_\mode)\couriercap_\mode$, we can define the pick-up time as the expected travel time of the nearest courier:
	\begin{equation}\label{eq:pick-up time}
		\plij(x)=\frac{\minutes_\mode}{1+\bij\couriercap_\mode(1-\rho_\mode)},
	\end{equation}
	where we used the fact that the expected minimum value of $n$ independent uniform random variables in $[0,1]$ is $\frac{1}{n+1}$.\par 
	
	Before continuing, we want to note that although our latency function $\lij(x)$ defined in Eq.~\eqref{eq:click-to-door} does not account for traffic congestion on the road, this is not an inherent limitation of our model. Specifically, we choose to set the travel time $\tlij$ as constant in order to focus our latency model on the congestion caused by courier utilization, as captured by the pick-time $\plij$ defined in Eq.~\eqref{eq:pick-up time}. Since we expect couriers to be a small percentage of total traffic, we believe this simplification to be justified. Nonetheless, given that our theoretical results in Section~\ref{sec: Theoretical Framework} make no restrictions to the latency function $\lij$, and the fact that our current formulation of latency $\lij$ is non-linear and non-convex, one can consider a model for the travel time $\tlij$ which accounts for congestion.\par
	
	\subsection{Cost Model}
	Before setting up our optimization problem, we need to model the cost of operating this system. We define the average dollar cost of completing a single order using a courier of modality $\mode$ as the delivery cost $\opcj$. This way, we can define the total cost of running our delivery system given the allocation strategy $x$:
	\begin{equation}\label{eq:total cost 1}
		\bopcost(x) = \sum_{\mode=1}^{\numModes}\opcj\big(\sum_{\order\in\Orders}\xij\big),
	\end{equation}
	where $\bopcost(x)$ is units of dollars per hour because $\xij$ is a rate of orders per hour. We use this model in our first case study concerning the meal delivery service. Alternatively, one can define a cost model using fixed hourly wages $\bar{c}_\mode$ for different courier modalities $\mode$, making the total cost of our system independent of the allocation strategy $x$:
	\begin{equation}\label{eq:total cost 2}
		\bopcost = \sum_{\mode=1}^{\numModes}\bar{c}_\mode\couriercap_\mode,
	\end{equation}
	where $\bar{c}_\mode$ is now in units of dollars per hour. Whereas we only account for completed orders in Eq.~\eqref{eq:total cost 1}, by modeling couriers using average wages in Eq.~\eqref{eq:total cost 2} we are accounting for the unused couriers. We use this model in our second case study concerning the taxi service as information regarding wages is readily available. In practice, more sophisticated cost models can be utilized by addressing statistics such as profit margins, travel distance for couriers, and other information that is available to the service provider.\par
	
	\subsection{Optimization Problem}
	We are now ready to set up the overall optimization problem. 
	\begin{align}		
		\min_{x}&\qquad \tlat(x)=\frac{1}{|\Orders|}\sum_{\order\in\Orders}\sum_{\mode\in\Modes}\lij(x)\xij\label{eq:final_optimization}\\
		\textrm{subject to }&\qquad \bopcost(x) \leq \sum_{\order\in\Orders}\sum_{\mode\in\Modes}\tij\xij,\label{eq:con_cost}\\
		&\qquad  \rho_\mode(x)\leq\Bar{\rho} \qquad \forall \mode\in\Modes,\label{eq:con_server}\\
		&\qquad \sum_{\mode\in\Modes}\xij=1,\qquad \forall\order\in\Orders,\label{eq:con_flow}\\
		&\qquad 0\leq\xij\leq1,\qquad \forall\order\in\Orders,\mode\in\Modes\label{eq:con_bounds}.
	\end{align}
	For this case study, we want to find the allocation strategy $x$ which minimizes expected latency $\tlat$, as shown in Eq.~\eqref{eq:final_optimization}. In addition, we constrain the operational cost in Eq.~\eqref{eq:con_cost} to be less than the total compensation received from all deliveries. Note that because Theorem~\eqref{thm:pricing condition} allows us to arbitrarily set price for the cheapest delivery option, one can always set the minimum allotted price to satisfy the constraint after optimization. We also constrain courier utilization in Eq.~\eqref{eq:con_server} by choosing an appropriate upper bound $\Bar{\rho}$ for all modalities $\mode\in\Modes$. Note that with this additional supply constraint, we assume that the number of available couriers $\couriercap_\mode$ is large enough so that the set of possible solutions is nonempty, as required by Proposition~\ref{prop: ex_nash}. We use the constraint in Eq.~\eqref{eq:con_flow} to satisfy demands for each order $\order\in\Orders$. Finally, we bound our decision variable between the domain of $[0,1]$ in Eq.~\eqref{eq:con_bounds} so that there are no negative values in the solution. Once we find a desirable allocation strategy by solving this optimization problem, we can set prices using our theoretical results such that this flow is induced at equilibrium.\par
	
	The optimization problem defined above is non-linear and non-convex, and we use a public implementation of the interior-point filter line-search algorithm~\cite{wachter2006implementation} to solve it, noting that this method can be used to solve nonlinear programs on the order of a million variables~\cite{ipopt_effic}. As aforementioned, many choices can be made for the formulation of the latency functions $\lij$, cost constraint $\bopcost$, and the optimization objective $\tlat$. To efficiently use the interior point method, it is desired for the objective function and constraints to be twice differentiable so that the Hessian can be defined. We include details regarding this implementation in \ref{sec: app_implementation}, and provide our code online~\cite{my_implementation}.\par 
	
	
	\section{Case Studies}\label{sec: Case Studies}
	We can now discuss the setting and results in our case studies. We first model a meal delivery system with three transportation modalities: cars, delivery drones, and and sidewalk autonomous delivery robots (SADRs). In this setting, we model each population corresponding to order $\order$ with a trade-off function $\pop_{\order}$ that is independent of the modalities provided. Our second setting considers a taxi service transporting customers to and from the Chicago O'Hare International Airport (ORD) using three transportation modalities: cars, luxury cars, and electric vertical takeoff and landing (eVTOL) aircrafts. Unlike the first study where users are simply ordering food, we now model our populations with different trade-off functions $\pop_{\order,\mode}$ for each modality. We list how the problem parameters are defined and our results below, providing the full implementation in our code online~\cite{my_implementation}.\par
	
	\subsection{Meal Delivery: Grubhub Instance}
	\noindent\textbf{Setup.}
	To define the problem parameters for our optimization formulation, we used real world instances from Grubhub~\cite{reyes2018meal}, which list information about the orders placed and car couriers available throughout a given time interval. Although there is no consideration of other modalities, we use the provided information as a basis and define our remaining parameters to be consistent. For service time $\slij$, we directly used the given pickup and dropoff times. Similarly for travel time $\tlij$, we used the provided distances between restaurant and customer locations, converting them to time by using constant speeds for all modalities. For cars, we set the speed to $11.93$ mph according to the dataset. For drones, we set the speed to $30$ mph, using the upper end of the reported range of $13-34$ mph~\cite{MARCINA2020} since food deliveries are relatively light. Finally for SADRs, we set the speed to $4$ mph, as they are expected to operate at the typical speed of pedestrians~\cite{SADRspeed}.\par 
	
	To calculate pickup time $\plij$, we computed all the parameters required in Eq.~\eqref{eq:pick-up time}. The number of couriers $\couriercap$ was directly chosen for each instance so that the problem was feasible under the utilization capacities $\Bar{\rho}_\mode$. For cars and SADRs, we set $\Bar{\rho}_\mode$ to $0.9$, while for drones, we decreased this value to $0.8$ due to the smaller number of vehicles utilized. We then generated courier locations for all three modalities, and computed the portion of available couriers $\bij$ that were at most $\minutes=10$ minutes away from the restaurant corresponding to order $\order$. For car couriers, we directly sampled from the provided locations, while for drone couriers, we sampled uniformly from a grid spanning the restaurant locations. To capture SADRs delivering from restaurants closer to downtown, we sampled their locations uniformly from a grid centered in the middle of all restaurant locations, with length and width equal to their coordinate's respective standard deviations. Using these parameters, we estimated the mean rate $\drate_\mode$ of order completions as the inverse of expected latency $\mathbb{E}_{i}[\lij]^{-1}$ for each modality $\mode$, assuming load was equally distributed across them.\par
	
	To compute the operational costs $\opcj$ for each modality, we set the cost per order to $\$5$ for car and drone deliveries as they are expected to be competitive under certain regimes~\cite{MckinseyDronesToSky}, while using a lower cost per order of $\$1.50$ for SADRs as the current cost per order is expected to drop from $\$2$ to $\$1$ in the near future~\cite{SADRcost}. For user trade-off between price and time $\pop(a)$, we used a linear function with the lowest evaluation $\pop(1)$ set to $\$10$ per hour, and the highest $\pop(0)$ set to $\$100$ per hour, for all orders $\order\in\Orders$. We go on to discuss the results of our case study for an instance with $505$ unique orders, each demanded with an equal rate of approximately $0.54$ deliveries per hour.\par
	
	\begin{table}[!h]
		\centering
		\caption{Results for a meal delivery system with $100$ car couriers when there is only one modality available. The total operational cost is $\$1,375$ per hour.}
		\label{tab:result_1}
		\centering
		\begin{tabular}{lc}
			\toprule
			& Cars \\
			\cmidrule(r){1-1}\cmidrule(r){2-2}\\
			Orders $(\%)$ & $100$ \\
			Utilization $\rho_j$ $(\%)$ & $88$\\
			Latency $\lat_j$ (min) & $19$ \\
			Distance (miles) & $1.47$ \\
			Price $\tax_j$ $(\$)$ & $5.00$ \\
			\bottomrule 
		\end{tabular}
	\end{table}
	
	\begin{table}[!h]
		\centering
		\caption{Results for a meal delivery system with $20$ car, $24$ drone, and $100$ SADR couriers available, with a total operational cost of $\$856.84$ per hour.}
		\label{tab:result_2}
		\centering
		\begin{tabular}{lcccc}
			\toprule
			& Cars & Drones & SADRs & Total\\
			\cmidrule(r){1-1}\cmidrule(r){2-4}\cmidrule(r){5-5}
			Orders $(\%)$ & $17$ & $29$ & $54$ & $100$\\
			Utilization $\rho_j$ $(\%)$ & $90$ & $80$ & $88$ & $86$\\
			Latency $\lat_j$ (min) & $24$ & $15$ & $26$ & $23$\\
			Distance (miles) & $1.69$ & $2.42$ & $0.88$ & $1.47$\\
			Price $\tax_j$ $(\$)$ & $4.13$ & $7.08$ & $0.65$ & $3.12$\\
			\bottomrule 
		\end{tabular}
		
	\end{table}
	\noindent\textbf{Results.}
	We first consider the case when there are only $100$ car couriers available, with no other transportation modality. We show the result in Table~\ref{tab:result_1}, listing the portion of orders delivered and the courier utilization as percentages, the average latency $\lat_j$ for modality $\mode$ in minutes, the distance between customer and restaurant in miles, and finally the delivery price in dollars. For this setting, prices were set so that the money accrued from delivery services was equal to the operational cost. Since car couriers have an operational cost of $\$5.00$ per order, we need an average delivery price of $\$5.00$ per order to satisfy it. Note that although in this setting the minimum delivery price can be set arbitrarily for all orders since users have no choice to make, when we introduce other delivery modalities this is no longer the case as the prices must follow Eq.~\eqref{eq:main_result} to satisfy Nash conditions.\par 
	
	Next, we consider the case when there are $20$ car, $24$ drone, and $100$ SADR couriers available displayed in Table~\ref{tab:result_2}. With the introduction of low cost SADRs into the system, the average latency has increased from $18$ to $23$ minutes, while the average price has decreased from $\$5.00$ to $\$3.12$ per order. Although a high cost and fast delivery service is now available via drones, customers are expected to pay a premium. These trends are expected for two reasons: (1) SADRs are much cheaper to operate making the total operational cost smaller, and (2) the faster speed of drones allows us to charge users who favor shorter delivery times more than users who favor cheaper delivery prices. It is interesting to note that drones are used to complete orders for customers furthest away from their chosen restaurant, while SADRs are used for customers closest to their chosen restaurant. This is due to the travel speeds of the different transportation modalities, as drones can travel efficiently between distant destinations, while robots are expected to operate in a smaller range as they have a pedestrian pace.\par 
	
	Overall, this case study shows that by setting prices according to users' trade-offs between money and time, one can implement a desired allocation strategy over multiple delivery modalities while improving their profit margins. One flaw with the study above is that we assumed users in different neighborhoods had identical distributions for their value of time. Although we had no information that would allow us to model this discrepancy, we expect that this would have an effect on where drones are utilized, as wealthier customers may reside in areas that are closer to downtown locations. Such an effect is considered in the following case study, where we model users' value of time (VOT) based on previous data.\par 
	
	\subsection{Taxi Service with Urban Air Transportation: Chicago O'Hare International Airport}
	
	\noindent\textbf{Setup.}
	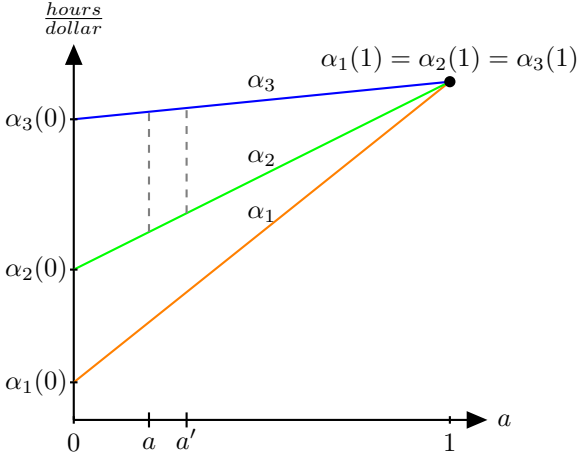
\begin{figure}
		\caption{The above schematic is an example of the user preference functions $\pop_{\order,\mode}$ for three modes $\Modes=\{1,2,3\}$, where the index corresponding to order $\order\in\Orders$ has been left out for convenience. The $x$-axis represents users $\user\in[0,1]$ and the $y$-axis represents hours per dollar, the inverse of the value of time. We can see that the three lines all intersect at $\user=1$, signifying that the most frugal users do not care about the modality and prefer the cheapest option. All three lines have varying slopes, and we sort their indices so that for $\user<\user'$, $\frac{\pop_{\order,\mode+1}(\user)}{\pop_{\order,\mode}(\user)}\geq \frac{\pop_{\order,\mode+1}(\user')}{\pop_{\order,\mode}(\user')}$ for all $\mode\in\{1,\ldots,\numModes-1\}$. The dashed gray lines help visualize this property: as we increase $\user$ to $1$, the ratio $\frac{\pop_{3}(\user)}{\pop_{2}(\user)}$ decreases due to $\pop_2$ having a larger slope than $\pop_3$, reaching a minimum value of $1$ when $\user=1$. Intuitively, this orders modalities by their slopes in decreasing order, representing a decrease in relative luxury. For our setting, this order is eVTOL aircrafts, luxury vehicles, and standard vehicles.}
		\begin{tikzpicture}[scale=5]
			\draw[->, thick] (0,0) -- (1.1,0) node[right]{$\user$};
			\draw[-, thick, black]  (1,.01) -- (1,-0.01) node[below]{$1$};
			\draw[-, thick, black]  (0,.01) -- (0,-0.01) node[below]{$0$};
			\draw[->, thick] (0,0) -- (0,1) node[above]{$\frac{hours}{dollar}$};
			
			\draw[-, thick, orange] (0,0.1) -- (1,0.9) node[midway, above, black]{$\pop_1$};
			\draw[-, thick, black]  (-0.01,0.1) -- (0.01,0.1) node[left]{$\pop_1(0)$};
			
			\draw[-, thick, green] (0,0.4) -- (1,0.9) node[midway, above, black]{$\pop_2$};
			\draw[-, thick, black]  (-0.01,0.4) -- (0.01,0.4) node[left]{$\pop_2(0)$};
			
			\draw[-, thick, blue]  (0,0.8) -- (1,0.9) node[midway, above, black]{$\pop_3$};
			\draw[-, thick, black]  (-0.01,0.8) -- (0.01,0.8) node[left]{$\pop_3(0)$};
			
			\node at (1, 0.9) [circle, fill, inner sep=1.5pt]{};
			\draw (1,0.9) node[above] {$\pop_1(1)=\pop_2(1)=\pop_3(1)$};
			
			\draw[-, thick, gray, dashed]  (0.2,0.5) -- (0.2,0.82);
			\draw[-, thick, gray, dashed]  (0.3,0.55) -- (0.3,0.83);
			\draw[-, thick, black]  (0.2,.02) -- (0.2,-0.02) node[below]{$\user$}; 
			\draw[-, thick, black]  (0.3,.02) -- (0.3,-0.02); 
			\node at (0.3, -0.05) {$\user'$};
			
		\end{tikzpicture}
		\label{fig:alpha schematic}
	\end{figure}
	For the second setting, we consider passengers requesting rides to and from the ORD Airport via three available transportation modalities: cars, luxury cars, and eVTOL aircrafts. Using the publicly available data provided by rideshare companies in the city of Chicago~\cite{TNP2023}, we analyzed all taxi requests between January 2023 and March 2023, collecting travel times, taxi fares, as well as pickup and dropoff locations. The dataset is conveniently divided into $77$ city areas, with the vertiport and ORD airport located at area $\#31$ and area $\#75$, respectively. Accordingly, we set the number of unique orders in our problem to $153$, representing all trips to and from the airport.\footnote{We assumed that trips which start and end in area $\#75$ were departing at the airport.} We set the total demand roughly equal to the average rate of $1,550$ orders per hour reported in the dataset, and distributed this proportionally among the orders. 
	
	For travel time $\tlij$ by car, we used the mean of all reported travel times between respective city areas. When considering trips via eVTOL aircrafts, $10$ minutes of flight time was added to the travel time required to get to and from the vertiport by car~\cite{archer}. To calculate pickup time $\plij$, we computed all the parameters required in Eq.~\eqref{eq:pick-up time}. The number of couriers $\couriercap$ was directly chosen for each instance so that the problem was feasible under the utilization capacities $\Bar{\rho}_\mode$, with the additional consideration that each eVTOL aircrafts would carry $4$ passengers~\cite{archer}. For both car modalities we set $\Bar{\rho}_\mode$ to $0.95$, while for eVTOL aircrafts we decreased this value to $0.75$ due to the smaller number of vehicles utilized. To estimate the portion of available car taxis near each city area $\bij$, we assumed that the higher densities of drivers were located in areas with higher demand. To model this, we exponentially scaled $\bij$ between $0.1$ and $0.3$ according to the demand at the corresponding city area. We set $\minutes$ to $5$ minutes for both car modalities, and $0$ minutes for eVTOL aircrafts as they do not need to travel to the pickup location. Service time $\slij$ was set to $0$ minutes for all modalities. Using these parameters, we estimated the mean rate $\drate_\mode$ of order completions as the inverse of expected latency $\mathbb{E}_{i}[\lij]^{-1}$ for each modality $\mode$, assuming load was equally distributed across them. Based on estimated driver wages in Chicago reported by Uber~\cite{UberPay}, we set the hourly wages $\bar{c}_\mode$ for standard cars and luxury cars to $\$25$ and $\$40$ per hour, respectively\footnote{The reported average wage was $\$28.03$ per hour for all drivers}. For eVTOL aircrafts we assumed an operational cost of $\$200$ per hour for each vehicle, meaning $\$50$ for each passenger.\par 
	
	Unlike the previous case study, we rely on Eq.~\ref{eq:extended_main} to derive our prices. Hence, in addition to modeling the heterogeneity in individuals' value of time (VOT), we assumed that more luxurious transportation modes are more valuable to the user. Previous studies have shown that users are willing to pay more for air taxi transportation compared to ground transportation~\cite{FU2019, binder_2018}. To represent this, for each order $\order$ we computed the mean VOT based on the dataset, and scaled it by $2, 1.5,$ and $1$ for eVTOL aircrafts, luxury cars, and standard cars, respectively. Furthermore, we assumed that the most frugal users at $\user=1$ would prefer the cheapest option regardless of the modality, and set $\pop_{\order,\mode}(1)$ for all modes $\mode$ equal to one standard deviation below the computed mean VOT for that order $\order$. To connect these two points for each order $\order$ and modality $\mode$, we used a straight line, making our overall model linear with larger slopes representing more luxurious modalities. We visualize this in Fig.~\ref{fig:alpha schematic}. Note that because of the common intercepts at $\pop_{\order,\mode}(1)$ for all modalities $\mode$ and the constant slopes, the assumption in Corollary~\ref{corr: ex_network_flow} is satisfied. Furthermore, by not offering users travel via eVTOL aircrafts when they are slower than the other two options, we can strengthen Corollary~\ref{corr: ex_network_flow} by guaranteeing that any desired flow is indeed inducible under some well behaved equilibrium flow. This result is shown in \ref{sec: app_case_study}. We make this aforementioned restriction in our formulation, and also explicitly check that the Nash equilibrium conditions of our solutions are indeed satisfied under the derived prices.\par
	
	\noindent\textbf{Results.}
	\begin{table}[!t]
		\centering
		\caption{Results for the airport taxi system with $990$ standard vehicles and $110$ luxury vehicles, with a total operational cost of $\$23,100$ per hour.}
		\label{tab:result_3}
		\centering
		\begin{tabular}{lccc}
			\toprule
			& Standard & Luxury & Total\\
			\cmidrule(r){1-1}\cmidrule(r){2-3}\cmidrule(r){4-4}
			Orders $(\%)$ & $91$ & $9$ & $100$\\
			Utilization $\rho_j$ $(\%)$ & $94$ & $85$ & $93$\\
			Latency $\lat_j$ (min) & $35$ & $35$ & $35$\\
			Price $\tax_j$ $(\$)$ & $41.35$ & $60.63$ & $43.06$\\
			Profit $(1,000\$$ per hour) & $38.60$ & $5.04$ & $43.59$\\
			\bottomrule 
		\end{tabular}
	\end{table}
	
	\begin{table}[!t]
		\centering
		\caption{Results for the airport taxi system with $900$ standard vehicles, $100$ luxury vehicles, and $25$ eVTOL aircrafts, with a total operational cost of $\$26,000$ per hour.}
		\label{tab:result_4}
		\centering
		\begin{tabular}{lcccc}
			\toprule
			& Standard & Luxury & eVTOL & Total\\
			\cmidrule(r){1-1}\cmidrule(r){2-4}\cmidrule(r){5-5}
			Orders $(\%)$ & $83$ & $8$ & $9$ & $100$\\
			Latency $\lat_j$ (min) & $35$ & $30$ & $23$ & $34$\\
			Utilization $\rho_j$ $(\%)$ & $94$ & $87$ & $75$ & $92$\\
			Price $\tax_j$ $(\$)$ & $40.96$ & $42.50$ & $112.15$ & $47.26$\\
			Profit $(1,000\$$ per hour) & $34.80$ & $2.38$ & $10.08$ & $47.26$\\
			\bottomrule 
		\end{tabular}
	\end{table}
	We first consider the case when there are only $1100$ car couriers available, with $110$ luxury vehicles and $990$ standard vehicles. We show the result in Table~\ref{tab:result_3}, listing the portion of orders delivered and the courier utilization as percentage, the average latency $\lat_j$ for modality $\mode$ in minutes, the delivery price in dollars, and finally the total profit in $1000$'s of dollars per hour. For this setting, the minimum price was set so that users taking a standard taxi would pay their expected amount. Recall that setting the minimum price does not alter the Nash equilibrium since demand is inelastic, and hence we set it accordingly to reflect current prices experienced by users. We see this reflected in the results, as customers are expected to pay $\$41.35$ for a standard cab, and $\$60.63$ for a luxury cab on average. We expect this since both vehicles have identical travel times, but the mean VOT for luxury vehicles $50\%$ higher compared to standard vehicles. This shows the importance of modeling VOT separately for different transportation modes, as such an effect would not be observed otherwise.\par  
	
	Next we consider the case when $25$ eVTOL aircrafts are introduced into the system, displayed in Table~\ref{tab:result_4}. As expected, we can see that this premium transportation option comes with a high cost of $\$112.15$ per order. Although services provided by standard vehicles are mostly unaffected, demand for the luxury vehicle option is severely impacted. Specifically, luxury vehicles have a harder time competing and are required to lower their price from $\$60.63$ to $\$42.50$. In addition, we see their average trip latency drop from $35$ to $30$ minutes, signifying that commuters located further away are more willing to pay the price for the eVTOL aircraft option. With the higher costs required to operate luxury vehicles, we can expect such competition to greatly affect profits provided by this service, most likely lowering the expected wage for luxury vehicle drivers.\par
	
	This case study outlines the importance of modeling VOT separately for different transportation modalities, as this can point out potential pitfalls in market strategies as new transportation modalities are introduced into the rideshare ecosystem. As we observed, even though eVTOL aircrafts only took $9\%$ of the total orders, they had a large affect on the profitability of already existing luxury transportation options.\par 
	
	\section{Conclusion}\label{sec: Conclusion}
	We model the pickup and delivery problem with multiple transportation modalities as a congestion game played over a star network, and show that we can explicitly define prices to induce any desired network flow. With this framework, we construct case studies for both a meal delivery service and a taxi service. In the first setting, we show that by utilizing autonomous transportation methods which are more efficient, one can set prices according to users' trade-offs between money and time to induce a desired allocation strategy while improving their profit margins. The second setting considers the additional assumption that users' trade-offs may differ between transportation modalities, and shows that such consideration are crucial to predict trends as new transportation modalities are introduced. We go over some of the implications of our work, pointing out limitations and directions for improvement.\par
	
	We first note that in the setting of non-atomic congestion games taking place on graphs composed of one source-sink pair, prior works have asked if a feasible solution can be found to compute optimal prices for edges combinatorially, without relying on LP formulations~\cite{roughgarden2003a}. Our main theoretical result states that in these settings, one can define optimal prices for paths combinatorially, implying that the LP formulation used to find prices for edges can be simplified. This points to the possibility that other network structures inherit properties which allow one to find prices efficiently, and we leave this direction for future works.\par
	
	Further, we point out that our case studies only provide two examples where such a model is useful. Due to the general construction of the congestion game defined, our analysis is practical for any application that utilizes a platform to price match customers with different transportation methods. Since our formulation poses little restriction on the latency function defined, one can construct a model that is suitable for the desired application. Of course our framework gives no guarantees on finding the optimal allocation strategy, and instead provides a method by which prices can be set to induce a desired strategy.\par 
	
	One direction for future work is relaxing our formulation to allow for elastic demand. Such a consideration is interesting as it would permit one to optimize for profits directly, since the minimum price set would now affect the total user demand. The difficulty of such an extension lies in the ordering permutations required to compute the delivery prices. In this setting, additional assumptions may be required as one would need to compute derivative information while keeping track of ordering permutations that depend directly on the decision variable.\par 
	
	Lastly, we want to comment on the ethical implications of our first case study. On the positive side, our results show that by utilizing more autonomous transportation methods one can improve profit margins. However, this is true because our model considers car couriers operated by humans as less cost efficient. While one may have financial incentives to substitute part of their current workforce with autonomous machines, other decisions can be made that improve wages and work conditions for employees. Such a discussion is beyond the scope of our work, and is a topic that should be carefully addressed by policy makers before corporations are allowed to make decision that greedily improve their profits.\par
	
	\section{Acknowledgements}
	This work was supported by the National Science Foundation [award numbers 1952920 and 2145134].
	
	\bibliographystyle{IEEEtran}
	\bibliography{IEEEabrv,refs.bib}
	\newpage
	\appendix
	\section{Proof of Proposition 2}\label{sec: app_prop_proof}
	As stated by prior work, the proof of Proposition~\ref{prop: ex_canon} (Proposition 2.4 in~\cite{roughgarden2003a}) can be arrived at using a rearrangement argument, showing the stronger statement that an arbitrary Nash flow can be reorganized into a canonical one without changing the disutility incurred by any agent or the induced flow on paths. As this proof is ommited from the aforementioned work, we provide an independent result for completeness.\par 
	
	Our proof strategy is as follows: assuming that $x$ is an arbitrary Nash flow, we use the inequalities defined by a Nash flow in Eq.~\ref{eq:Nash_condition} to show that the additional inequalities defined by a canonical Nash flow in Definition~\ref{def:canon} must follow. The existence of canonical Nash flows follows directly from the well known result stated in Proposition~\ref{prop: ex_nash}. Note that for sake of notation, we drop the subscript referring to orders $\order\in\Orders$, as it should be clear that the proof applies to an individual source-sink pair.\par
	
	Formally, $x$ is an equilibrium flow for instance $(\pop,\lat,\tax)$ if for all edges $\mode\in\setmodes$, no user $\user$ travelling on edge $\mode$ should want to switch to any other edge $\mode'\in\setmodes$:
	\begin{equation}\label{eq:proof_prop_1}
		\lat_\mode+\pop(\user)\tax_\mode \leq \lat_{\mode'}+\pop(\user)\tax_{\mode'} \quad\forall\user\in\{\user:x(\user)=\mode\},
	\end{equation}
	where we leave out denoting the flow $x$ in latency $\lat_\mode(x)$. Given an arbitrary Nash flow $x$ for instance $(\pop,\lat,\tax)$, we would like to show that for any two users $\user_1,\user_2\in[0,1]$ where $\user_1<\user_2$: $\lat_{x(\user_1)}\leq\lat_{x(\user_2)}$, and $\tax_{x(\user_1)}\geq\tax_{x(\user_2)}$ must hold.\par 
	
	Clearly, when users $\user_1$ and $\user_2$ are routed on the same edge $\mode$, the aforementioned inequalities hold. To show that they hold in general, we assume that user $\user_1$ is routed on edge $\mode$, $\user_1\in\{\user:x(\user)=\mode\}$, and user $\user_2$ is routed on edge $\mode'$, $\user_2\in\{\user:x(\user)=\mode'\}$, where $\mode\neq\mode'$. It follows directly from Eq.~\eqref{eq:proof_prop_1} that:
	
	\begin{equation}\label{eq:proof_prop_2}
		\lat_{\mode}-\lat_{\mode'} \leq \pop(\user_1)(\tax_{\mode'}-\tax_{\mode}),
	\end{equation}
	and similarly, 
	\begin{equation}\label{eq:proof_prop_3}
		\lat_{\mode}-\lat_{\mode'} \geq \pop(\user_2)(\tax_{\mode'}-\tax_{\mode}).
	\end{equation}
	
	Given $\pop(\user)\geq0$ $\forall\user\in[0,1]$ and $\pop(\user_1)\leq\pop(\user_2)$ from definition, we can infer from the above two inequalities that $\tax_{\mode}\geq\tax_{\mode'}$ and $\lat_{\mode}\leq\lat_{\mode'}$. This is trivially shown by contradiction: assume that $\tax_{\mode'}-\tax_{\mode}$ is positive, and divide by it on both sides of the inequalities to arrive at the contradiction $\pop(\user_2)\leq\frac{\lat_{\mode}-\lat_{\mode'}}{\tax_{\mode'}-\tax_{\mode}}\leq\pop(\user_1)$, implying that $\tax_{\mode'}-\tax_{\mode}$ and $\lat_{\mode}-\lat_{\mode'}$ are negative. Since $\lat_{x(\user_1)}\leq\lat_{x(\user_2)}$ and $\tax_{x(\user_1)}\geq\tax_{x(\user_2)}$ must hold for any two users $\user_1,\user_2\in[0,1]$ where $\user_1<\user_2$, it follows directly that for any edge $\mode\in\Modes$, the users assigned to edge $\mode$ by a flow $x$ at Nash equilibrium form a (potentially empty or degenerate) subinterval of $[0,1]$. This completes the proof.  
	\section{Proof of Theorem 1}\label{sec: app_proof}
	Note that for sake of notation, we will drop the subscript referring to orders $\order\in\Orders$, as it should be clear that the proof applies to an individual source-sink pair. In addition, we assume that indexes $\mode\in\Modes:\{1,\ldots,\numModes\}$ correspond to the set of edges sorted by non-decreasing latency. Note that throughout our proof, we apply Proposition~\ref{prop: ex_canon} which allows us to assume that any Nash flow $x$ is a canonical Nash flow, as demonstrated in ~\ref{sec: app_prop_proof} above.\par
	
	We define two adjacent intervals that are formed by our flow $x$: users $\user\in[\user_{\mode-1},\user_\mode]$ on the left experience latency $\lat_\mode$ and price $\tax_\mode$, while users $\user\in[\user_{\mode},\user_{\mode+1}]$ on the right experience latency $\lat_{\mode+1}$ and price $\tax_{\mode+1}$. The two intervals are portrayed in Fig.~\ref{fig:proof intervals}, where we note that this definition holds for $\mode\in\{1,\ldots,\numModes-1\}$. Using the inequalities defined in Eq.~\eqref{eq:Nash_condition}, we know that for $x$ to be a (canonical) Nash flow for instance $(\pop,\lat,\tax)$, no user $\user$ from the left interval $\user\in[\user_{\mode-1},\user_\mode]$ should want to switch to the delivery option corresponding to the right interval:
	
	\begin{equation}\label{eq:left_to_right}
		\lat_\mode+\pop(\user)\tax_\mode \leq \lat_{\mode+1}+\pop(\user)\tax_{\mode+1} \quad\forall\user\in[\user_{\mode-1},\user_\mode],
	\end{equation}
	where we leave out denoting the flow $x$ in latency $\lat_\mode(x)$. It follows:
	
	\begin{align}
		\tax_\mode -\tax_{\mode+1} &\leq \frac{\lat_{\mode+1}-\lat_{\mode}}{\pop(\user)} \quad\forall\user\in[\user_{\mode-1},\user_\mode],\\
		\tax_\mode -\tax_{\mode+1} &\leq \min_{\user\in[\user_{\mode-1},\user_\mode]}\Big({\frac{\lat_{\mode+1}-\lat_{\mode}}{\pop(\user)}}\Big).
	\end{align}
	
	The preceding inequality can be simplified further by using the non-decreasing property of function $\pop$ defining the population's price sensitivity: for any $\user_1,\user_2\in[0,1]$ such that $\user_1\leq\user_2$, given user $\user\in[\user_1,\user_2]$, $\max{\pop(\user)}=\pop(\user_2)$ and $\min{\pop(\user)}=\pop(\user_1)$. This comparison results in the following condition which must be true for $x$ to be a Nash flow:   
	\begin{equation}\label{eq:left_ineq}
		\tax_\mode -\tax_{\mode+1} \leq {\frac{\lat_{\mode+1}-\lat_{\mode}}{\pop(\user_\mode)}}.
	\end{equation}
	
	We can repeat this process by enforcing that no user $\user$ from the right interval $\user\in[\user_{\mode},\user_{\mode+1}]$ should want to switch to the edge on the left:
	
	\begin{align}
		\lat_{\mode+1}+\pop(\user)\tax_{\mode+1} &\leq \lat_{\mode}+\pop(\user)\tax_{\mode} \quad\forall\user\in[\user_{\mode},\user_{\mode+1}],\\
		\tax_\mode -\tax_{\mode+1} &\geq \max_{\user\in[\user_{\mode},\user_{\mode+1}]}\Big({\frac{\lat_{\mode+1}-\lat_{\mode}}{\pop(\user)}}\Big),
	\end{align}
	which results in the following:
	
	\begin{equation}\label{eq:right_ineq}
		\tax_\mode -\tax_{\mode+1} \geq {\frac{\lat_{\mode+1}-\lat_{\mode}}{\pop(\user_\mode)}}.
	\end{equation}
	
	From \eqref{eq:left_ineq} and \eqref{eq:right_ineq} we can see that the two inequalities force the set of prices $\settax$ to follow:
	\begin{equation}\label{eq:proof_eq}
		\tax_\mode -\tax_{\mode+1} = {\frac{\lat_{\mode+1}-\lat_{\mode}}{\pop(\user_\mode)}} \quad\forall \mode\in\{1,\ldots,\numModes-1\},
	\end{equation}
	where if $\tax_{\numModes}$ is given, the rest of the prices can be found recursively as defined in Eq~\eqref{eq:main_result}.\par
	
	To complete the proof, we must show that for this set of prices $\tax$, the desired $x$ is indeed an equilibrium flow. Formally, $x$ is an equilibrium flow for instance $(\pop,\lat,\tax)$ if for all edges $\mode\in\setmodes$ no user $\user$ in interval $\user\in[\user_{\mode-1},\user_\mode]$ should want to switch to any other edge $\mode'\in\setmodes$:
	\begin{equation}\label{eq:proof_all_ineq}
		\lat_\mode+\pop(\user)\tax_\mode \leq \lat_{\mode'}+\pop(\user)\tax_{\mode'}.
	\end{equation}
	Clearly, these inequalities hold when $\mode=\mode'$, and hence we show that they hold when $\mode>\mode'$ and $\mode<\mode'$. Starting with the former, when $\mode>\mode'$ we are considering that no user choosing edge $\mode$ will switch to any edge $\mode'$ on the left, where by definition $\tax_\mode\leq\tax_{\mode'}$ and $\lat_\mode\geq\lat_{\mode'}$. Rearranging Eq.~\ref{eq:proof_all_ineq}, we have the following for all edges $\mode>\mode'$:
	\begin{align*}
		&\lat_\mode+\pop(\user)\tax_\mode \leq \lat_{\mode'}+\pop(\user)\tax_{\mode'} \quad\forall\user\in[\user_{\mode-1},\user_\mode],\\
		&\tax_{\mode'} -\tax_\mode \geq \max_{\user\in[\user_{\mode-1},\user_{\mode}]}\Big(\frac{\lat_\mode-\lat_{\mode'}}{\pop(\user)}\Big),\\
		&\sum_{k=\mode'}^{\numModes-1}\frac{\lat_{k+1}-\lat_{k}}{\pop(\user_k)}-\sum_{k=\mode}^{\numModes-1}\frac{\lat_{k+1}-\lat_{k}}{\pop(\user_k)}\geq \frac{\lat_\mode-\lat_{\mode'}}{\pop(\user_{\mode-1})},\\
		&\sum_{k=\mode'}^{\mode-1}\frac{\lat_{k+1}-\lat_{k}}{\pop(\user_k)}\geq \sum_{k=\mode'}^{\mode-1}\frac{\lat_{k+1}-\lat_{k}}{\pop(\user_{\mode-1})}.
	\end{align*}
	Since $\pop(\user_{\mode-1})\geq\pop(\user_k)$ when $\mode'\leq k \leq \mode-1$, every summation term on the left hand side is strictly greater than or equal to every summation term on the right hand side, validating the inequalities in Eq.~\ref{eq:proof_all_ineq} for $\mode>\mode'$. We can do the same for $\mode<\mode'$, where now $\tax_\mode\geq\tax_{\mode'}$ and $\lat_\mode\leq\lat_{\mode'}$:
	\begin{align*}
		&\lat_\mode+\pop(\user)\tax_\mode \leq \lat_{\mode'}+\pop(\user)\tax_{\mode'} \quad\forall\user\in[\user_{\mode-1},\user_\mode],\\
		&\tax_\mode -\tax_{\mode'} \leq \min_{\user\in[\user_{\mode-1},\user_{\mode}]}\Big(\frac{\lat_{\mode'}-\lat_\mode}{\pop(\user)}\Big),\\
		&\sum_{k=\mode}^{\numModes-1}\frac{\lat_{k+1}-\lat_{k}}{\pop(\user_k)}-\sum_{k=\mode'}^{\numModes-1}\frac{\lat_{k+1}-\lat_{k}}{\pop(\user_k)}\leq \frac{\lat_\mode-\lat_{\mode'}}{\pop(\user_{\mode})},\\
		&\sum_{k=\mode}^{\mode'-1}\frac{\lat_{k+1}-\lat_{k}}{\pop(\user_k)}\leq \sum_{k=\mode}^{\mode'-1}\frac{\lat_{k+1}-\lat_{k}}{\pop(\user_{\mode})}.
	\end{align*}
	This time, since $\pop(\user_{\mode})\leq\pop(\user_k)$ when $\mode\leq k \leq \mode'-1$, every summation term on the left hand side is strictly less than or equal to every summation term on the right hand side. This completes the proof.
	
	\begin{remark}
		Although Equation~\ref{eq:main_result} in Theorem~\ref{thm:pricing condition} defines prices for parallel edges, this is equivalent to finding prices for paths in general directed graphs composed of one source-sink pair. Briefly consider a directed graph $G=(V,E)$ with source $s$ and sink $t$, with edges $e\in E$ and simple $s-t$ paths $P\in\mathcal{P}$. Since the desired flow $x:[0,1]\mapsto\mathcal{P}$ is provided, the path latency can be found using $\ell_P(x)=\sum_{e\in P}\ell_e(x_e)$, where the edge flow is given by $x_e=\sum_{P:e\in P}x_P$. The proofs for Proposition 2 and Theorem 1 stated above follow directly by replacing edges $\mode\in\Modes$ with paths $P\in\mathcal{P}$. Note that this does not guarantee that there exists a unique set of additive edge prices $\tax_e$ such that $\tax_P=\sum_{e\in P}\tax_e$ is true for all paths $P\in\mathcal{P}$. Since users choosing between delivery modes can be represented by parallel edges, we forego defining paths in our formulation to be concise.\end{remark}
	

	\section{Proof of Corollary 2}\label{sec: app_corr_proof}
	The proof strategy is as follows: similar to the proof of Theorem~\ref{thm:pricing condition}, using a subset of the inequalities defined for Nash equilibrium in Eq.~\eqref{eq:Nash_condition}, we show that for some desired Nash flow $\setflow$ there is only one set of valid prices $\tax_\order$ that satisfies this subset of inequalities. Note that we drop the subscript referring to orders $\order\in\Orders$, as it should be clear that the proof applies to an individual source-sink pair. In addition, we assume that indexes $\mode\in\Modes:\{1,\ldots,\numModes\}$ are ordered to satisfy the following assumption:
	
	\begin{equation}\label{eq: assumption_alpha}
		\textrm{Given }\user\leq\user': 
		\frac{\pop_{\mode+1}(\user)}{\pop_{\mode}(\user)}\geq \frac{\pop_{\mode+1}(\user')}{\pop_{\mode}(\user')} \qquad \forall\mode\in\{1,\ldots,\numModes-1\},
	\end{equation}
	where $\pop$ is a set of non decreasing functions $\pop_{\mode}:[0,1]\rightarrow(0,\infty)$ for $\mode\in\setmodes$.\par
	
	As before, we define two adjacent intervals that are formed by our desired flow $x$: users $\user\in[\user_{\mode-1},\user_\mode]$ on the left experience latency $\lat_\mode$ and price $\tax_\mode$, while users $\user\in[\user_{\mode},\user_{\mode+1}]$ on the right experience latency $\lat_{\mode+1}$ and price $\tax_{\mode+1}$. The two intervals are no longer guaranteed to exhibit the properties of a canonical Nash flow portrayed in Fig.~\ref{fig:proof intervals}. Nonetheless, using the inequalities defined in Eq.~\eqref{eq:Nash_condition}, we know that for $x$ to be a Nash flow for instance $(\pop,\lat,\tax)$, no user $\user$ from the left interval $\user\in[\user_{\mode-1},\user_\mode]$ should want to switch to the delivery option corresponding to the right interval:
	
	\begin{equation}\label{eq:corr_left_to_right}
		\lat_\mode+\pop_{\mode}(\user)\tax_\mode \leq \lat_{\mode+1}+\pop_{\mode+1}(\user)\tax_{\mode+1} \quad\forall\user\in[\user_{\mode-1},\user_\mode],
	\end{equation}
	where we leave out denoting the flow $x$ in latency $\lat_\mode(x)$. It follows:
	
	\begin{align}
		\tax_\mode -\frac{\pop_{\mode+1}(\user)}{\pop_{\mode}(\user)}\tax_{\mode+1} &\leq \frac{\lat_{\mode+1}-\lat_{\mode}}{\pop_{\mode }(\user)} \quad\forall\user\in[\user_{\mode-1},\user_\mode],\\
		\max_{\user\in[\user_{\mode-1},\user_\mode]}\Big(\tax_\mode -\frac{\pop_{\mode+1}(\user)}{\pop_{\mode}(\user)}\tax_{\mode+1}\Big) &\leq \min_{\user\in[\user_{\mode-1},\user_\mode]}\Big(\frac{\lat_{\mode+1}-\lat_{\mode}}{\pop_{\mode }(\user)}\Big).
	\end{align}
	The LHS of the above inequality can be simplified using the assumption made in Eq.~\ref{eq: assumption_alpha}, whereas the LHS can be simplified as before by using the non-decreasing property of function $\pop_\mode$ defining the population's price sensitivity. This results in the following condition which must be true for $x$ to be a Nash flow:   
	
	\begin{equation}\label{eq:corr_left_ineq}
		\tax_\mode -\frac{\pop_{\mode+1}(\user_\mode)}{\pop_{\mode}(\user_\mode)}\tax_{\mode+1} \leq \frac{\lat_{\mode+1}-\lat_{\mode}}{\pop_{\mode }(\user_\mode)}.
	\end{equation}
	
	We can repeat this process by enforcing that no user $\user$ from the right interval $\user\in[\user_{\mode},\user_{\mode+1}]$ should want to switch to the edge on the left:
	
	\begin{align}
		\lat_\mode+\pop_{\mode}(\user)\tax_\mode &\geq \lat_{\mode+1}+\pop_{\mode+1}(\user)\tax_{\mode+1} \quad\forall\user\in[\user_{\mode},\user_{\mode+1}],\\
		\min_{\user\in[\user_{\mode},\user_{\mode+1}]}\Big(\tax_\mode - \frac{\pop_{\mode+1}(\user)}{\pop_{\mode}(\user)}\tax_{\mode+1}\Big) &\geq \max_{\user\in[\user_{\mode},\user_{\mode+1}]}\Big({\frac{\lat_{\mode+1}-\lat_{\mode}}{\pop_\mode(\user)}}\Big),
	\end{align}
	
	which results in the following:
	
	\begin{equation}\label{eq:corr_right_ineq}
		\tax_\mode -\frac{\pop_{\mode+1}(\user_\mode)}{\pop_{\mode}(\user_\mode)}\tax_{\mode+1} \geq \frac{\lat_{\mode+1}-\lat_{\mode}}{\pop_{\mode }(\user_\mode)}.
	\end{equation}
	
	From \eqref{eq:corr_left_ineq} and \eqref{eq:corr_right_ineq} we can see that the two inequalities force the set of prices $\settax$ to follow:
	\begin{equation}\label{eq:corr_proof_eq}
		\tax_\mode -\frac{\pop_{\mode+1}(\user_\mode)}{\pop_{\mode}(\user_\mode)}\tax_{\mode+1} = \frac{\lat_{\mode+1}-\lat_{\mode}}{\pop_{\mode }(\user_\mode)} \quad\forall \mode\in\{1,\ldots,\numModes-1\},
	\end{equation}
	where if $\tax_{\numModes}$ is given, the rest of the prices can be found recursively. This completes the proof.\par
	
	\begin{remark}
		Note that unlike Theorem~\ref{thm:pricing condition}, we can not generalize the result of Corollary~\ref{corr: varying modality} to apply for any desired flow. Since we can not rely on Proposition~\ref{prop: ex_canon} to order the edges as we do for a canonical Nash flow, we can not show that $\tax_\order$ does indeed satisfy all of the inequalities defined in Eq.~\eqref{eq:Nash_condition}. Hence we only prove that for a desired flow $\setflow$ to be an equilibrium flow, the prices $\settax$ must follow Eq.~\ref{eq:corr_proof_eq}.
	\end{remark}
	
	\section{Generalization of Corollary~\ref{corr: varying modality}}\label{sec: app_case_study}
	
	As mentioned in Section~\ref{sec: Case Studies}, by not offering users travel via eVTOL aircrafts when they are slower than the other two options, we can guarantee that any desired flow is indeed inducible by a well behaved equilibrium flow. More precisely, we assume that there are three modalities $\numModes=3$, and the most premium option $\mode=1$ must have the smallest latency $\lat_{\order,1}\leq\lat_{\order,2},\lat_{\order,3}$ for all considered orders $\order$. Given this, we can show that all of the inequalities defined for Nash equilibrium in Eq.~\eqref{eq:Nash_condition} can be satisfied under some choice of the cheapest price $\tax_{\order,3}$, which is free for us to define. We proceed to show this, dropping the subscript referring to orders $\order\in\Orders$ as it should be clear that the proof applies to an individual source-sink pair.\par
	
	First recall that the result of Corollary~\ref{corr: varying modality} directly satisfies the inequalities defined for Nash equilibrium when adjacent edges are considered. In otherwords, the prices given by Eq.~\ref{eq:extended_main} guarantee that users will not switch from luxury cars $\mode=2$ to standard cars $\mode=3$ or vice-versa, as well as luxury cars $\mode=2$ to eVTOL aircrafts $\mode=1$ or vice-versa. This result is shown in \ref{sec: app_corr_proof} for the general setting. Since we only have three edges in this setting, two inequalities are left to check. No user $\user$ from the leftmost interval $\user\in[0,\user_1]$ corresponding to eVTOL aircrafts $\mode=1$ should want to switch the right most interval corresponding to standard cars $\mode=3$:
	
	\begin{equation}\label{eq:case_left_to_right}
		\lat_1+\pop_{1}(\user)\tax_1 \leq \lat_{3}+\pop_{3}(\user)\tax_{3} \quad\forall\user\in[0,\user_1],
	\end{equation}
	
	and vice-versa:
	
	\begin{equation}\label{eq:case_right_to_left}
		\lat_1+\pop_{1}(\user)\tax_1 \geq \lat_{3}+\pop_{3}(\user)\tax_{3} \quad\forall\user\in[\user_2,1].
	\end{equation}
	
	As done in \ref{sec: app_corr_proof}, the above two inequalities can be simplified to remove the intervals, and combined to form the following:
	
	\begin{equation}\label{eq:case_full_1}
		\frac{\pop_1(\user_1)}{\pop_3(\user_1)}\tax_1 - \frac{\lat_3-\lat_1}{\pop_3(\user_1)} \leq \tax_3 \leq \frac{\pop_1(\user_2)}{\pop_3(\user_2)}\tax_1 - \frac{\lat_3-\lat_1}{\pop_3(\user_2)}. 
	\end{equation}
	
	Note that since $\lat_3-\lat_1\geq0$ and $\pop_3(\user_1)\leq\pop_3(\user_2)$, we have  $-\frac{\lat_3-\lat_1}{\pop_3(\user_1)}\leq-\frac{\lat_3-\lat_1}{\pop_3(\user_2)}$. In addition, we have $\frac{\pop_1(\user_1)}{\pop_3(\user_1)}\leq\frac{\pop_1(\user_2)}{\pop_3(\user_2)}$ due to the assumption stated in Corollary~\ref{corr: varying modality}. This means that the domain of possible values for $\tax_3$ is not degenerate, meaning we can set $\tax_3$ to satisfy both Eq.~\ref{eq:case_left_to_right} and Eq.~\ref{eq:case_right_to_left} simultaneously. In practice, we set $\tax_3$ to reflect current prices as mentioned in Section~\ref{sec: Case Studies}, and verify that it is within the permissible range.\par 
	
	\section{Implementation Details}\label{sec: app_implementation}
	We use a public implementation of the interior-point filter line-search algorithm~\cite{wachter2006implementation}, and integrate it with the dataset of Grubhub instances~\cite{reyes2018meal} and Chicago taxi services~\cite{TNP2023} using Python. We briefly outline the results needed to implement the algorithm, and provide our code online~\cite{my_implementation}.
	
	First we define derivative information for the latency, where we use the fact that $\frac{d\rho_j(x)}{dx_{i',j'}}=\frac{1_{[j=j']}}{\couriercap_j\drate_j}$.
	\begin{align}
		\lij &= \slij + \tlij + \minutes_j[1+\bij\couriercap_j(1-\rho_j(x))]^{-1},\\
		\frac{d\lij}{d\xijp} &= 1_{[j=j']}\frac{\bij\minutes_j}{\drate_j}[1+\bij\couriercap_j(1-\rho_j(x))]^{-2},\\
		\frac{d^2\lij}{d\xijpp d\xijp} &= 1_{[j=j'=j'']}\frac{2\bij^2\minutes_j}{\drate_j^2}[1+\bij\couriercap_j(1-\rho_j(x))]^{-3}.
	\end{align}
	
	Now we can use this to get derivative information for the objective function.
	
	\begin{align}
		\tlat(x) &= \frac{1}{|\Orders|}\sum_{\order\in\Orders}\sum_{\mode\in\Modes}\lij\xij,\\
		\frac{d\tlat(x)}{d\xijp} &= \frac{1}{|\Orders|}\Bigg[\Big(\sum_{\order\in\Orders}\sum_{\mode\in\Modes}\frac{d\lij}{d\xijp}\xij\Big) + \lijp\Bigg],\\
		\frac{d^2\tlat(x)}{d\xijpp d\xijp} &= \frac{1}{|\Orders|}\Bigg[\Big(\sum_{\order\in\Orders}\sum_{\mode\in\Modes}\frac{d^2\lij}{d\xijpp d\xijp}\xij\Big) + \frac{d\lijpp}{d\xijp} + \frac{d\lijp}{d\xijpp}\Bigg].
	\end{align}
	
	Next, we derive derivative information for the two constraints separately, starting with the utility constraints in Eq.~\eqref{eq:con_server} which we denote as $g_j(x)\leq\Bar{\rho} \quad\forall \mode\in\Modes$. 
	
	\begin{align}
		g_\mode(x) &= \rho_\mode(x),\\
		\frac{dg_\mode(x)}{d\xijp} &= \frac{1_{[j=j']}}{\couriercap_j\drate_j},\\
		\frac{d^2g_\mode(x)}{d\xijpp d\xijp} &= 0.
	\end{align}
	
	Finally, we denote the flow constraints in Eq.~\eqref{eq:con_flow} as $h_i(x)=1 \quad\forall \order\in\Orders$, listing the derivative information below.
	
	\begin{align}
		h_\order(x) &= \sum_{\mode\in\Modes}\xij,\\
		\frac{dh_\order(x)}{d\xijp} &= 1_{[i=i']},\\
		\frac{d^2h_\order(x)}{d\xijpp d\xijp} &= 0.
	\end{align}
	
	Note that the cost constraint in Eq~\eqref{eq:con_cost} is not required for optimization as the minimum price can be manually set to satisfy it.
\end{document}